\documentclass[11pt]{article}

\usepackage{graphicx,fullpage,verbatim,color}
\usepackage{amsmath}
\usepackage{amssymb}
\usepackage{float}

\floatstyle{ruled}
\newfloat{algorithm}{thp}{lop}

\newenvironment{proof}{\noindent\textbf{Proof}}{\hfill\qed}
\newcommand{\qed}{\hfill$\Box$}
\newtheorem{lemma}{Lemma}
\newtheorem{theorem}{Theorem}
\newtheorem{definition}{Definition}
\newtheorem{specification}{Specification}

\title{The Byzantine Brides Problem}

\author{
Swan Dubois\protect\footnote{UPMC Sorbonne Universit\'{e}s \& Inria, France, swan.dubois@lip6.fr} 
\and 
S{\'e}bastien Tixeuil\protect\footnote{UPMC Sorbonne Universit\'{e}s \& Institut Universitaire de France, France, sebastien.tixeuil@lip6.fr} 
\and 
Nini Zhu\protect\footnote{UPMC Sorbonne Universit\'{e}s, France}}

\date{}

\begin{document}

\maketitle

\begin{abstract}
We investigate the hardness of establishing as many stable marriages (that is, marriages that last forever) in a population whose memory is placed in some arbitrary state with respect to the considered problem, and where traitors try to jeopardize the whole process by behaving in a harmful manner. On the negative side, we demonstrate that no solution that is completely insensitive to traitors can exist, and we propose a protocol for the problem that is optimal with respect to the traitor containment radius. 
\end{abstract}

\section{Introduction}

After 1123 years of existence, the Byzantine Empire finally collapsed soon after the fall of Constantinople in 1453 by the Ottoman army (see Figure~\ref{constantinople}). The various wars that opposed armies in the previous years ravaged their homeland as well as the capital city, as a contemporary reported~\cite{B53b}: ``The blood flowed in the city like rainwater in the gutters after a sudden storm.'' 

\begin{figure}[htbp]
\centering
\includegraphics[height=5cm]{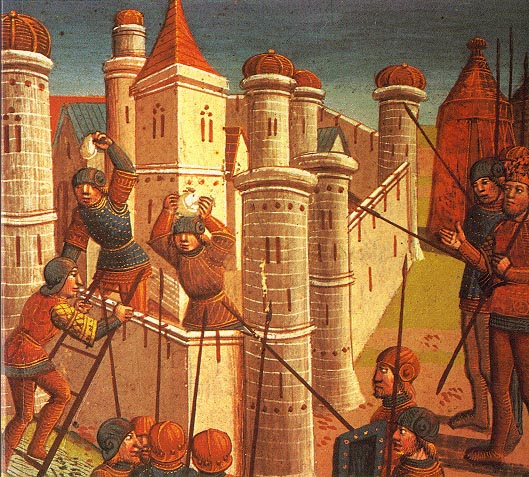}
\caption{Scene from the battle defending Constantinople, Paris 1499}
\label{constantinople}
\end{figure}

Allegedly, the main reason for the Byzantine defeat is that there were traitors amongst its leading generals~\cite{PSL80j,LSP82j}. With traitors at their cores, armies suffered significant losses, leaving mostly widows, orphans, and devastated homes. After the country was taken and the truce signed, the city was to rebuild, starting with its core roots: families. In the ancient days, strict guidelines were followed to form new marriages, like coming from the same social circles or being of opposite sex. In a wasted land with few homes still standing, those were no longer sustainable options. Stability of marriages was decided to be the most important criterium, rendering every other consideration irrelevant. So, general guidelines were to be followed by all survivors: \emph{(i)} do your best to make your marriage last, \emph{(ii)} don't be picky about whom you are married to, and \emph{(iii)} don't make others' marriage fail. Still, the Byzantine traitors that led the armies to their doom were hidden amongst the surviving population, and managed somehow to remain unnoticed. Their purpose was to cause as much havoc as possible, by any means necessary, without being caught for their socially inconvenient behavior. So, the reconstruction of the city could have been jeopardized by few nasty Byzantine brides or bridegrooms.  

The core problem Byzantine authorities were facing to establish as many stable marriages as possible lied in the following two observations:
\begin{enumerate}
\item the population was heavily shocked by the war that just stopped, and their state of mind was somewhat erratic: some could not remember they were previously married, some though they were previously married but never were, some though they were engaged and expected a response that would never come because the engagement was not remembered by the expected bride or bridegroom, etc,
\item the traitors could simulate emotional shock in order to stay undiscovered yet try to perturbate the global marriage process.
\end{enumerate}
So, the only difference between the general population and the traitors was their willingness to accommodate the stable marriage doctrine in their daily life.

In this paper, we investigate the hardness of establishing as many stable marriages (that is, marriages that last foverer) in a population whose memory is placed in some arbitrary state with respect to the considered problem, and where traitors try to jeopardize the whole process by behaving in a harmuful manner. On the negative side, we demonstrate that no solution that is completely insensitive to traitors can exist, and we propose a protocol for the problem that is optimal with respect to the traitor containment radius. 

\section{Model and Definitions}

\subsection{State Model}

A \emph{Byzantine city} $S=(V,L)$ consists of a set $V=\{v_1,v_2,\ldots,v_n\}$ of potential brides\footnote{Note that we use the word ``bride'' in the sequel of this paper to denote both brides and bridegrooms.} (or simply brides) and a set $L$ of potential marriages. A potential marriage is an unordered pair of distinct potential brides (this takes place before the Internet ages, so long distance marriage is not supposed to last forever, and only marriages occurring in a vicinity may be stable). A Byzantine city $S$ can be regarded as a graph whose vertex set is $V$ and whose link set is $L$, so in the sequel we use graph terminology to describe a Byzantine city $S$. We use the following notations: $n=|V|$, $m=|L|$ and $d(u,v)$ denotes the distance between two nodes $u$ and $v$ (\emph{i.e} the length of the shortest path between $u$ and $v$).

Potential brides $u$ and $v$ are called \emph{neighbors} if $(u,v)\in L$. The set of neighbors of a potential bride $v$ is denoted by $N_v$. We do not assume existence of unique identifiers for potential brides (Birth records have been destroyed by the war, and memory of each potential bride is unreliable). Instead we assume each potential bride may distinguish its neighbors from each other by locally labeling them.

For the sake of generality and the lack of reports concerning the remains of Constantinople after it has fallen, we consider that the Byzantine city has arbitrary yet connected topology. We adopt the \emph{shared state model} \cite{D74j} as a communication model, where each potential bride can directly and instantaneously get the current status of its neighbors.

The current memory that is maintained by a potential bride is denoted by the term of state, and may be further divided into variables. A potential bride may take actions that are prescribed by the authorities during the reconstruction of the Byzantine city. An action is simply a function that is executed in an atomic manner by the potential bride. The action executed by each potential bride is described by a finite set of guarded commands of the form $\langle$guard$\rangle\longrightarrow\langle$statement$\rangle$. Each guard of potential bride $u$ is a Boolean expression involving the state of $u$ and its neighbors.

A global state of a Byzantine city is called a \emph{configuration} and is specified by the product of states of all potential brides. We define $C$ to be the set of all possible configurations of a Byzantine city $S$. For a potential bride set $R \subseteq V$ and two configurations $\gamma$ and $\gamma'$, we denote $\gamma \stackrel{R}{\mapsto} \gamma'$ when $\gamma$ changes to $\gamma'$ by executing an action of each potential bride in $R$ simultaneously. Notice that $\gamma$ and $\gamma'$ can be different only in the states of potential brides in $R$. For completeness of execution semantics, we should clarify the configuration resulting from simultaneous actions of neighboring potential brides. The action of a potential bride depends only on the current state at $\gamma$ and the states of the neighbors at $\gamma$, and the result of the action reflects on the state of the potential bride at $\gamma '$.

We say that a potential bride is \emph{enabled} in a configuration $\gamma$ if the guard of at least one of its actions evaluates as true in $\gamma$. A \emph{schedule} of a Byzantine city is an infinite sequence of potential bride sets. Let $Q=R^1, R^2, \ldots$ be a schedule, where $R^i \subseteq V$ holds for each $i\ (i \ge 1)$. An infinite sequence of configurations $e=\gamma_0,\gamma_1,\ldots$ is called an \emph{execution} from an initial configuration $\gamma_0$ by a schedule $Q$, if $e$ satisfies $\gamma_{i-1} \stackrel{R^i}{\mapsto} \gamma_i$ for each $i\ (i \ge 1)$. Potential bride actions are executed atomically, and we distinguish some properties on the scheduler (or daemon). A \emph{distributed daemon} schedules the actions of potential brides such that any subset of potential brides can simultaneously execute their actions. We say that the daemon is \emph{central} if it schedules action of only one potential bride at any step. The set of all possible executions from $\gamma_0\in C$ is denoted by $E_{\gamma_0}$. The set of all possible executions is denoted by $E$, that is, $E=\bigcup_{\gamma\in C}E_{\gamma}$. We consider \emph{asynchronous} Byzantine cities but we add the following assumption on schedules: any schedule is central and fair (meaning that only one enabled potential bride is chosen at any step and that no potential bride can be infinitely often enabled without being chosen by the scheduler)

In this paper, we consider (permanent) \emph{Byzantine faults}: a Byzantine potential bride (\emph{i.e.} a Byzantine-faulty potential bride) can exhibit arbitrary behavior independently of its actions. If $v$ is a Byzantine-faulty potential bride, $v$ can repeatedly change his (or her) state arbitrarily. For a given execution, the number of faulty potential brides is arbitrary.

\subsection{Self-Stabilizing Protocols Resilient to Byzantine Faults}

As the problem we solve is meant for stability and should reach a global fixed point, we use a \emph{specification predicate} (shortly, specification) to define it. This specification predicate is denoted by $spec(v)$, for each potential bride $v$. A configuration is a desired one if every potential bride satisfies $spec(v)$. A specification $spec(v)$ is a Boolean expression on variables of $P_v~(\subseteq P)$ where $P_v$ is the set of potential brides whose state (or part of) appear in $spec(v)$. The variables appearing in the specification are called \emph{output variables} (shortly, \emph{O-variables}). 

A \emph{self-stabilizing protocol} (\cite{D74j,D00b,T09bc}) is a protocol that eventually reaches a \emph{legitimate configuration}, where $spec(v)$ holds at every potential bride $v$, regardless of the initial configuration. Once it reaches a legitimate configuration, every potential bride never changes its O-variables and always satisfies $spec(v)$. From this definition, a self-stabilizing protocol is expected to recover from any number and any type of transient faults. However, the recovery from any configuration is guaranteed only when every potential bride honestly executes its action from the configuration, \emph{i.e.}, self-stabilization does not consider the possibility of Byzantine-faulty potential brides.

When (permanent) Byzantine-faulty potential brides exist, they may not satisfy $spec(v)$. In addition, honest potential brides near the Byzantine-faulty potential brides can be influenced and may be unable to satisfy $spec(v)$. Nesterenko and Arora~\cite{NA02c} define a \emph{strictly stabilizing protocol} as a self-stabilizing protocol resilient to unbounded number of Byzantine-faulty actors. 

\begin{definition}($c$-honest potential bride)
A potential bride is $c$-honest if it is honest (\emph{i.e.} not Byzantine-faulty) and located at distance more than $c$ from any Byzantine-faulty potential bride.
\end{definition}

\begin{definition}($(c,f)$-containment)
\label{def:cfcontained}
A configuration $\gamma$ is \emph{$(c,f)$-contained} for specification $spec$ if, given at most $f$ Byzantine-faulty potential brides, in any execution starting from $\gamma$, every $c$-honest potential bride $v$ always satisfies $spec(v)$ and never changes its O-variables.
\end{definition}

The parameter $c$ of Definition~\ref{def:cfcontained} refers to the \emph{containment radius} defined by Nesterenko and Arora~\cite{NA02c}. The parameter $f$ refers explicitly to the number of Byzantine-faulty potential brides, while \cite{NA02c} dealt with an arbitrary number of Byzantine faults (that is, $f\in\{0\ldots n\}$).

\begin{definition}($(c,f)$-strict stabilization)
\label{def:cfstabilizing}
A protocol is \emph{$(c,f)$-strictly stabilizing} for specification $spec$ if, given at most $f$ Byzantine-faulty potential brides, any execution $e=\gamma_0,\gamma_1,\ldots$ contains a configuration $\gamma_i$ that is $(c,f)$-contained for $spec$.
\end{definition}

A specification is $r$-restrictive~\cite{NA02c} if it prevents combinations of states that belong to two potential brides $u$ and $v$ that are at least $r$ hops away. An important consequence for our purpose is that the containment radius of protocols solving $r$-restrictive specifications is at least $r$. 

\section{Specification}

The problem of maximal marriage construction is a well known problem in Distributed Computing. Given a graph $G=(V,E)$, a marriage $M$ on $G$ is a subset of $E$ such that any node of $V$ belongs to at most one edge of $M$. A marriage is maximal if there exists no marriage $M'$ such that $M\subsetneq M'$.

\begin{specification}(Maximal Marriage)~
\begin{description}
\item{Liveness:} The protocol terminates in a finite time.
\item{Safety:} In the terminal configuration, there exists a maximal marriage 
\end{description}
\end{specification}

Each potential bride $v$ has a variable $pref_v$ which belongs to the set $N_v\cup \{null\}$. This variable refers to the preferred neighbor of $v$ for a marriage. For example, if $pref_v=u$ then $v$ wants to add the edge $\{v,u\}$ to the marriage. For any potential bride $v$, we define the following set of predicates over the Byzantine city: \emph{(i)} $proposing_v$ denotes the fact that $v$ is proposing marriage to some neighbor $u$, but that $u$ has not shown interest yet, \emph{(ii)} $married_v$ denotes that $v$ has proposed $u$ and $u$ has proposed $v$ back, \emph{(iii)} $doomed_v$ denotes that $v$ has proposed neighbor $u$, but $u$ has proposed somebody else than $v$, \emph{(iv)} $dead_v$ denotes that $v$ has no hope of getting married ever (all neighbors proposed to somebody else), and \emph{(v)} $single_v$ means that $v$ has not proposed anyone and has at least one neighbor likewise. Formally:

\[\begin{array}{ccc}
proposing_v&\equiv&(pref_v=u)\wedge(pref_u=null)\\
married_v&\equiv&(pref_v=u)\wedge(pref_u=v)\\
doomed_v&\equiv&(pref_v=u)\wedge(pref_u=w)\wedge(w\neq v)\\
dead_v&\equiv&(pref_v=null)\wedge(\forall u\in N_v,married(u)=true)\\
single_v&\equiv&(pref_v=null)\wedge(\exists u\in N_v, married(u)\neq true)
\end{array}\]

It is easy to verify that for any configuration $\gamma$ and for any potential bride $v$, exactly one of these predicates holds for $v$ in $\gamma$.

If the Byzantine city is subject to Byzantine failures, obviously no protocol can satisfy the classical specification of the problem. Now, a potential bride $v$ is considered locally legitimate when it satisfies the following predicate: $spec(v)\equiv married_v\vee dead_v$. We now describe the global properties that are satisfied by a $(c,f)$-contained configuration for $spec$. Informally, we can prove that there exists a maximal marriage on a subset of $S$ in such a configuration and that this subset includes at least the set of $c$-honest potential brides. In the following, $V_c$ denotes the set of $c$-honest potential brides (\emph{i.e.}, $V_c=\{v\in V|\forall b\in B,d(v,b)>c\}$).

\begin{definition}($(c,\gamma)$-marriage subset)
Given an integer $c>0$ and a configuration $\gamma$, the $(c,\gamma)$-marriage subset $S^*_{c,\gamma}$ of $S$ is the subset induced by the following set of potential brides:
\[V'=V_c\cup\{v\in V\setminus V_c|\exists u\in V_c,pref_v=u\wedge pref_u=v\}\]
\end{definition}

Now, we can state formally the property satisfied by any $(c,f)$-contained configuration for $spec$.

\begin{lemma}
\label{lemma1}
In any $(c,f)$-contained configuration for $spec$, there exists a maximal marriage on the subset $S^*_{c,\gamma}$.
\end{lemma}

\begin{proof}
Let ${\gamma}$ be a $(c,f)$-contained configuration for $spec$. Hence, $\gamma$ satisfies $\forall v \in V_c$,$\ married_v\vee dead_v$. Let us define the following edge set $M_c=\{\{v, pref_v\}|v\in V_c\wedge pref_v\neq null\}$.

First, we show that $M_c$ is a marriage on $S^*_{c,\gamma}$. Indeed, if $\{v, pref_v\}$ is an edge of $M_c$, then $v$ satisfies $married_v $ (since $v$ satisfies $spec(v)$ and $pref_v\neq null$ by construction of $M_c$). Hence, we have $pref_{pref_v}=v$. Consequently, $v$ and $pref_v$ appear only once in $M_c$.

Now, we show that $M_c$ is maximal. By contradiction, assume it is not the case. Consequently, there exists two neighbors $v$ and $u$ (with $v\in V'$ and $u\in V'$) such that $\{v,u\}\notin M_c$ and $M_c'=M_c\cup\{\{v,u\}\}$ is a marriage on $S^*_{c,\gamma}$. Let us study the following cases: 

\begin{description}
\item[Case 1:] $u\in V_c$ and $v\in V_c$.\\
If $married_v\wedge married_u$ holds, then $\{v,u\}\in M_c$ by construction that contradicts the hypothesis. If $dead_v\wedge dead_u$ holds, then we can deduce that $(pref_v=null)\wedge(married_u)$ (since $dead_v$ holds), that contradicts $dead_u$. If $dead_v\wedge married_u$ (resp. $married_v\wedge dead_u$) holds, then $\{v,pref_v\}\in M_c$ with $pref_v\neq u$ (resp. $\{u,pref_u\}\in M_c$ with $pref_u\neq v$) and we can deduce that $v$ (resp. $u$) appears in two distinct edges of $M_c'$. Then, $M_c'$ is not a marriage that contradicts the hypothesis. 
\item[Case 2:] $u\notin V_c$ and $v\notin V_c$.\\
According to the assumption, $\{u,v\}\notin M_c$. Since $v\in V'\setminus V_c\wedge u\in V'\setminus V_c$, we have $\{v,pref_v\}\in M_c$ with $pref_v\neq u\wedge pref_v\in V_c$ (resp. $\{u,pref_u\}\in M_c$ with $pref_u\neq v\wedge pref_u\in V_c$) and we can deduce that $v$ (resp. $u$) appears in two distinct edges of $M_c'$. Then, $M_c'$ is not a marriage that contradicts the hypothesis.
\item[Case 3:] $u\in V_c$ and $v\notin V_c$.\\
According to the assumption, $\{v,u\}\notin M_c$. Since $v\in V'\setminus V_c\wedge u\in V_c$,we have $\{v,pref_v\}\in M_c$ with $pref_v\neq u\wedge pref_v\in V_c$ (since if $pref_v=u$, then $\{v,u\}\in M_c$ that contradicts the hypothesis) and we can deduce that $v$ appears in two distinct edges of $M_c'$. Then, $M_c'$ is not a marriage that contradicts the hypothesis.
\end{description}
\end{proof}

The result of Lemma~\ref{lemma1} motivates the design of a strictly stabilizing protocol for $spec$. Indeed, even if this specification is local, it induces a global property in $(c,f)$-contained configuration for $spec$ since there exists a maximal marriage of a well-defined sub-graph in such a configuration.

\section{Strictly Stabilizing Maximal Marriage}

This section presents our strictly stabilizing solution for the maximal marriage problem. We also prove its correctness and its optimality with respect to containment radius.

\subsection{Our Protocol}

Our strictly-stabilizing maximal marriage protocol, called $\mathcal{SSMM}$ is formally presented as Algorithm~\ref{algo1}. The basis of the protocol is the well-known self-stabilizing Maximal Marriage protocol by Huang and Hsu~\cite{HH92j}, but we allow potential brides to remember their past sentimental failures (\emph{e.g.} an aborted marriage du to the mate being Byzantine-faulty, or a proposal that didn't end up in an actual marriage) in order not to repeat the same mistakes forever when Byzantine-faulty brides participate to the global marriage process. The ideas that underly the marriage process for honest potential brides follows the directives discussed in the introduction: \emph{(i)} once married, honest brides never divorce and never propose to anyone else, \emph{(ii)} honest brides may propose to any neighbor, and if proposed, will accept marriage gratefully, \emph{(iii)} if they realize they previously proposed to somebody that is potentially married to somebody else, they will cancel their proposal and refrain proposing to the same potential bride soon. A potential bride $v$ maintain two variables: $pref_v$, that was already discussed in the problem specification section, and $old\_pref_v$ that is meant to recall past sentimental failures. Specifically, $old\_pref_v$ stores the last proposal made to a neighbor that ended up doomed (because that neighbor preferred somebody else, potentially because of Byzantine-faulty divorce, or because of genuine other interest that occurred concurrently). Then, the helper function $next\_v$ helps $v$ to move on with past failures by preferring the next mate not to be the same as previously (in a Round Robin order): the same potential bride that caused a sentimental breakup may be chosen twice in a row only if the only one available.

\begin{algorithm}
\caption{$\mathcal{SSMM}$: Strictly-stabilizing maximal marriage for potential bride $v$}
\label{algo1}
\begin{description}
\item[Variables:]~\\
$pref_v\in N_v\cup\{null\}$: preferred neighbor of $v$\\
$old\_pref_v\in N_v$: previous preferred neighbor of $v$
\item[Function:]~\\
For any $u\in\{v,null\}$, $next_v(u)$ is the first neighbor of $v$ greater than $old\_pref_v$ (according to a round robin order) such that $pref_{next_v(u)}=u$
\item [{Rules:}]~\\
/* Don't be picky: Accept any mate (round robin priority) */\\
$(M)::(pref_v=null)\wedge(\exists u\in N_v,pref_u=v)\longrightarrow pref_v:=next_v(v)$\\
/* Don't be picky: Propose to anyone (round robin priority) */\\
$(S)::(pref_v=null)\wedge(\forall u\in N_v,pref_u\neq v)\wedge(\exists u\in N_v, pref_u=null)\longrightarrow pref_v:=next_v(null)$\\
/* Don't cause others to break up: give up proposing if doomed */\\
$(A)::(pref_v=u)\wedge(pref_u\neq v)\wedge(pref_u\neq null)\longrightarrow old\_pref_v:=pref_v;pref_v:=null$
\end{description}
\end{algorithm}

\subsection{Proof of Strict Stabilization}

In their paper~\cite{HH92j}, Hsu and Huang prove the self-stabilizing property of their maximal marriage algorithm using a variant function. A variant function is a function that associates to any configuration a numerical value. This function is designed such that: \emph{(i)} the function is bounded, \emph{(ii)} any possible step of the algorithm decreases strictly the value of the function, and \emph{(iii)} the function reaches its minimal value if and only if the corresponding configuration is legitimate. Once such a function is defined and its properties are proved, we can easily deduce the convergence of the protocol. Indeed, whatever the initial configuration is, the associate value by the variant function is bounded (by property \emph{(i)}) and any execution starting from this configuration reaches in a finite time the minimal value of the function (by property \emph{(ii)}). Then, property \emph{(iii)} allows us to conclude on the convergence of the algorithm.

Our proof of strict-stabilization for our protocol also relies on a variant function (borrowed from the one of \cite{T94j}). We choose a variant function where we consider only potential brides of $V_2$. For any configuration $\gamma\in\Gamma$, let us define the following functions:
\[\begin{array}{rcl}
w(\gamma)&=&|\{v\in V_2|proposing_v\}|\\
c(\gamma)&=&|\{v\in V_2|doomed_v\}|\\
f(\gamma)&=&|\{v\in V_2|single_v\}|\\
P(\gamma)&=&(w(\gamma)+ c(\gamma)+ f(\gamma), 2c(\gamma) + f(\gamma))
\end{array}\]
Note that our variant function $P$ satisfies property \emph{(i)} by construction.

Then, we define the following configuration set:
\[\mathcal{LC}_2=\{\gamma\in\Gamma|\forall v\in V_2,spec(v)\}\]
In other words, $\mathcal{LC}_2$ is the set of configurations in which any potential bride $v$ of $V_2$ satisfies $spec(v)$.

We can now explain the road-map of our proof. After two preliminaries results (Lemmas~\ref{lemma2} and \ref{lemma3}) that are used in the sequel, we first show that any configuration of the set $\mathcal{LC}_2$ is $(2,n)$-contained for $spec$ (Lemma~\ref{lemma4}), that is, the set $\mathcal{LC}_2$ is closed by actions of $\mathcal{SSMM}$. Then, there remains to prove the convergence of the protocol to configurations of $\mathcal{LC}_2$ (starting from any configuration) to show the strict-stabilization of $\mathcal{SSMM}$. The remainder of the proof is devoted to the study of properties of our variant function $P$. First, we show in Lemma~\ref{lemma5} that any configuration $\gamma$ that satisfies $P(\gamma)=(0,0)$ belongs to $\mathcal{LC}_2$. This proves that $P$ satisfies the property \emph{(iii)}. Unfortunately, we can prove that our variant function $P$ does not satisfy property \emph{(ii)} (strict decreasing) since Byzantine faults may lead some potential brides to take actions that increase the function value. Nevertheless, we prove in Lemmas~\ref{lemma6}, \ref{lemma7}, and \ref{lemma8} that this case may appear only a finite number of times and that our variant function is eventually strictly decreasing, which is sufficient to prove the convergence to $\mathcal{LC}_2$ in Lemma~\ref{lemma9}. Finally, Lemmas~\ref{lemma4} and \ref{lemma9} permit to conclude with Theorem~\ref{theorem1} that establishes the $(2,n)$-strict stabilization of $\mathcal{SSMM}$. A detailed proof follows.

\begin{lemma}
\label{lemma2}
For any execution $e=\gamma_0,\gamma_1\ldots$,\\
\noindent - if $married_v$ holds in $\gamma_0$ for a potential bride $v\in V_1$, then $married_v$ holds in $\gamma_i$ for all $i\in \mathbb{N}$; and
\noindent - if $dead_v$ holds in $\gamma_0$ for a potential bride $v\in V_2$, then $dead_v$ holds in $\gamma_i$ for all $i\in \mathbb{N}$.
\end{lemma}

\begin{proof}
Let $v$ be a potential bride of $V_1$. Hence, any neighbor of $v$ is a honest potential bride. If $married_v$ holds in a configuration $\gamma_0$, then $pref_v=u\wedge pref_u=v$ holds in $\gamma_0$ by definition. We can observe that $v$ and $u$ are not enabled by $(M),(S)$, or by $(A)$ in $\gamma_0$. Consequently, $v$ and $u$ are not activated in any execution $e$ starting from $\gamma_0$. In conclusion, $married_v$ holds in any configuration of $e$.

Let $v$ be a potential bride of $V_2$. If $dead_v$ holds in a configuration $\gamma_0$, then $pref_v=null\wedge (\forall u\in N_v, married_u=true)$ holds in $\gamma_0$ by definition. Note that any neighbor of $v$ belongs to $V_1$ (since $v\in V_2$). If $married_u$ holds in $\gamma_0$, then $married_u$ holds in any configuration of any execution starting from $\gamma_0$. Potential Bride $v$ is not enabled by $\mathcal{SSMM}$ in $\gamma_0$. No neighbor of $v$ is enabled (according to the first part of the proof). Consequently, $dead_v$ holds in any configuration of any execution starting from $\gamma_0$.
\end{proof}

\begin{lemma}
\label{lemma3}
For any configuration $\gamma\in\mathcal{LC}_2$, no potential bride of $V_2$ is enabled by $\mathcal{SSMM}$ in $\gamma$.
\end{lemma}

\begin{proof}
Let ${\gamma}$ be a configuration of $\mathcal{LC}_2$. By definition, $\gamma$ satisfies $\forall v\in V_2, married_v\vee dead_v$. Let $v$ be a potential bride of $V_2$.

If $married_v$ holds in ${\gamma}$, then we have $pref_v=u$ and $pref_u=v$ by definition. We can observe that $v$ is not enabled by rules $(M)$ and $(S)$ in ${\gamma}$ since $pref\neq null$ and that $v$ is not enabled by rule $(A)$ in ${\gamma}$ since $pref_u=v$.

If $dead_v$ holds in ${\gamma}$, then we have $pref_v=null$ and $\forall u\in N_v, married_u=true$ by definition. We can observe that $v$ is not enabled by rule $(A)$ in ${\gamma}$ since $pref_v=null$ and that $v$ is not enabled by rules $(M)$ and $(S)$ in ${\gamma}$ since $\forall u\in N_v,married_u\Rightarrow\exists r_u,pref_u=r_u\neq v\neq null$.

In both case, $v$ is not enabled in ${\gamma}$. Hence, no potential bride of $V_2$ is enabled in ${\gamma}$.
\end{proof}

The definition of $\mathcal{LC}_2$ and Lemma~\ref{lemma2} allow us to state the following lemma:

\begin{lemma}
\label{lemma4}
Any configuration of $\mathcal{LC}_2$ is $(2,n)$-contained for $spec$.
\end{lemma}

\begin{lemma}
\label{lemma5}
Any configuration $\gamma\in\Gamma$ satisfying $P(\gamma)=(0,0)$ belongs to $\mathcal{LC}_2$.
\end{lemma}

\begin{proof}
If a configuration ${\gamma}\in\Gamma$ satisfies $P(\gamma)=(0,0)$, then $w(\gamma)+f(\gamma)+c(\gamma)=0$. Hence, no potential bride of $V_2$ is proposing, single, or doomed. Every potential bride $v$ of $V_2$ satisfies $married_v\vee dead_v$. By definition of $\mathcal{LC}_2$, we have ${\gamma}\in \mathcal{LC}_2$. 
\end{proof}

The following lemma is proved in a similar way as the corresponding one of \cite{T94j} (considering only potential brides of $V_2$).

\begin{lemma}
\label{lemma6}
For any configuration $\gamma\notin\mathcal{LC}_2$ and any step $\gamma\rightarrow\gamma'$ in which a potential bride of $V_2$ is activated by $\mathcal{SSMM}$, we have $P(\gamma')<P(\gamma)$.
\end{lemma}

\begin{proof}
Let $\gamma$ be a configuration such that ${\gamma}\notin \mathcal{LC}_2$. Consider any step $\gamma\rightarrow\gamma'$ of $\mathcal{SSMM}$. Since the scheduling is central, at most one potential bride $v\in V_2$ is activated during $\gamma\rightarrow\gamma'$. Consider the following cases.

\noindent\textbf{Case 1:} $v\in V_2$ is activated by rule $(M)$ during $\gamma\rightarrow\gamma'$.\\
By construction, there exists $u\in N_v$ such that $v$ and $u$ become married during this step. Hence, the function $w+f+c$ decreases by at least 1 during this step. Consequently, we have $P(\gamma')<P(\gamma)$.

\noindent\textbf{Case 2:} $v\in V_2$ is activated by rule $(S)$ during $\gamma\rightarrow\gamma'$.\\
As $pref_v=null$ and there exists $u\in N_v$ such that $pref_u=null$ in ${\gamma}$, $single_v$ holds in ${\gamma}$. On the other hand, $pref_v=u$ and $pref_u=null$ hold in ${\gamma'}$, that implies that $proposing_v$ holds in ${\gamma'}$. Hence, the function $2c+f$ decreases at least by one during $\gamma\rightarrow\gamma'$.

As rule $(S)$ is enabled in ${\gamma}$, we can deduce that no neighbor of $v$ is proposing after it (otherwise, the rule $(S)$ is not enabled for $v$). So no proposing node in ${\gamma}$ becomes single or doomed in ${\gamma'}$. If a neighbor of $v$ is single in ${\gamma}$, it remains single or become dead in ${\gamma'}$. We can conclude that the function $c+f+w$ remains equal and that the function $2c+f$ decreases by exactly one during $\gamma\rightarrow\gamma'$. Consequently, we have $P(\gamma')<P(\gamma)$.

\noindent\textbf{Case 3:} $v\in V_2$ is activated by rule $(A)$ during $\gamma\rightarrow\gamma'$.\\
As there exists $u\in N_v$ such that $pref_v=u$, $pref_u=w$, and $w\neq v$ in ${\gamma}$, we can deduce that $doomed_v$ holds in ${\gamma}$. As $pref_v=null$ holds in ${\gamma'}$, we know that $single_v\vee dead_v$ holds in ${\gamma'}$. Hence, the function $2c+f$ decreases by at least one during $\gamma\rightarrow\gamma'$.

If a neighbor of $v$ is single in ${\gamma}$, then it remains single in ${\gamma'}$ (note that $u$ cannot become dead in ${\gamma'}$ since $pref_v=null$). If a neighbor of $v$ is proposing in ${\gamma}$, it remains in this state because it cannot wait for $v$ (recall that $pref_v\neq null$ in ${\gamma}$). If a neighbor of $v$ is doomed to $v$ in ${\gamma}$, then it becomes proposing in ${\gamma'}$. Note that, if $u\in V_2$, $u$ leads to a supplementary decreasing of the function $2c+f$ while $c+w+f$ remains equal but, if $u\notin V_2$, then functions $2c+f$ and $c+w+f$ remains equal.  We can conclude that the function $2c+f$ decreases by at least one during $\gamma\rightarrow\gamma'$. Consequently, we have $P({\gamma'})<P({\gamma})$.
\end{proof}

\begin{lemma}
\label{lemma7}
In any execution, $P$ only increases a finite number of times.
\end{lemma}

\begin{proof}
Let $v$ be a potential bride of $V$. Let $e$ be an execution in which $P$ is not monotonically decreasing. Consider the first step $\gamma\rightarrow\gamma'$ of $e$ in which $P$ increases and in which $v$ is activated.

By Lemma~\ref{lemma6}, we know that $v\notin V_2$ (otherwise, we have a contradiction with the decreasing of $P$). Then, by construction of $P$, we know that $v\in V_1$ (since actions of potential brides of $V_0$ have no effects on values of $P$). Consequently, $v\in V_1\setminus V_2$.

Assume that $v$ executes rule $(S)$ during the step $\gamma\rightarrow\gamma'$. Observe that $v$ is $single$ in $\gamma$ but becomes $proposing$ in $\gamma$. Moreover, any neighbor of $v$ that is $single$ in $\gamma$ remains in this state in $\gamma'$ and there is no neighbor $dead$, $proposing$ after $v$, or $doomed$ after $v$ in $\gamma$ (by construction of the rule). By Lemma~\ref{lemma2}, any neighbor of $v$ in $V_2$ remains married in $\gamma'$ if it is $married$ in $\gamma$. Hence, the action of $v$ does not modify the state of its neighbors and $P$ is not modified, that contradicts the assumption.

Assume that $v$ executes rule $(A)$ during the step $\gamma\rightarrow\gamma'$. Observe that $v$ is $doomed$ in $\gamma$ but becomes $single$ in $\gamma'$. Moreover, any neighbor of $v$ that is $single$ in $\gamma$ remains in this state in $\gamma'$. By construction of the rule, there is no neighbor of $v$ that is $dead$ or $proposing$ after $v$ in $\gamma$. Any neighbor of $v$ that is $doomed$ after $v$ in $\gamma$ becomes $proposing$ in $\gamma'$. By Lemma~\ref{lemma2}, any neighbor of $v$ in $V_2$ remains married in $\gamma'$ if it is $married$ in $\gamma$. Hence, the action of $v$ leads to a strict decreasing of $P$, that contradicts the assumption.

Consequently, we know that $v$ is activated by rule $(M)$ during the step $\gamma\rightarrow\gamma'$. Then, by construction of the rule, we know that $v$ becomes $married$ in $\gamma'$ and remains in this state during the whole execution (by Lemma~\ref{lemma2}). In particular, $v$ is never activated in the sequel of the execution.

In conclusion, we obtain that each potential bride of $v\in V_1\setminus V_2$ executes at most one step that decreases $P$. As the number of potential bride of $v\in V_1\setminus V_2$ is finite, we obtain the result. 
\end{proof}

\begin{lemma}
\label{lemma8}
For any configuration $\gamma_0\notin\mathcal{LC}_2$ and any execution $e=\gamma_0,\gamma_1,\gamma_2,\ldots$ starting from $\gamma_0$, there exists a configuration $\gamma_i$ such that $P(\gamma_{i+1})<P(\gamma_i)$.
\end{lemma}

\begin{proof}
Let $\gamma_0$ be a configuration such that $\gamma_0\notin\mathcal{LC}_2$. By contradiction, assume that there exists an execution $e=\gamma_0,\gamma_1,\gamma_2,\ldots$ starting from $\gamma_0$ such that for any $i\in \mathbb{N}$, $P(\gamma_{i+1})\geq P(\gamma_i)$. By Lemma~\ref{lemma6}, that implies that no potential bride of $V_2$ is activated in any step of $e$.

As $\gamma_0\notin\mathcal{LC}_2$, there exists $v\in V_2$ such that $spec(v)$ does not hold in $\gamma_0$. Hence, $v$ is $proposing$, $doomed$, or $single$ in $\gamma_0$. Consider the following cases.

\noindent\textbf{Case 1:} $v$ is $proposing$ in $\gamma_0$. By definition, we have $\exists u\in N_v, (pref_v=u)\wedge (pref_u=null)$. 

\noindent\textbf{Case 1.1:} If $u\in V_2$, then we can observe that $u$ is enabled by $(M)$ in $\gamma_0$. Since $v$ and $u$ are never activated in $e$ (by Lemma~\ref{lemma6}), then $u$ remains continuously enabled. As the daemon is fair, $u$ is activated in a finite time, that is  contradictory.

\noindent\textbf{Case 1.2:} If $u\notin V_2$, then we can observe that $u$ is continuously enabled by $(M)$ from $\gamma_0$ (since $v$ is never activated). As the daemon is fair, $u$ executes $(M)$ in a finite time and becomes $married$. If $u$ is married with $v$, then $P$ decreases, that is contradictory. We can deduce that $u$ becomes $married$ with another potential bride and is never activated afterwards (by Lemma~\ref{lemma2}). Then, $v$ becomes $doomed$ and continuously enabled by rule $(A)$. As the daemon is fair, $v$ is activated in a finite time and becomes $dead$ or $single$. In both cases, $P$ decreases, that is contradictory.

\noindent\textbf{Case 2:} $v$ is $doomed$ in $\gamma_0$. By definition, we have $(pref_v=u)\wedge (pref_u=r)\wedge (r\neq v)$.
 
\noindent\textbf{Case 2.1:} If $u$ is activated in $e$, then we can observe that $v$ is continuously enabled by rule $(A)$. As the daemon is fair, $v$ is activated in a finite time, that is contradictory.

\noindent\textbf{Case 2.2:} If $u$ is activated in $e$, then we know that $u\notin V_2$ (otherwise, we obtain a contradiction). Before the first activation of $u$, we have $(pref_v=u)\wedge (pref_u=r)\wedge (pref_r=w)\wedge (r\neq v)\wedge (w\neq u)$ since the only enabled rule when $pref_u\neq null$ is $(A)$. After the execution of $(A)$ by $u$, $v$ becomes $proposing$ after $u$ and we can refer to case 1.2.

\noindent\textbf{Case 3:} If $single_v$ holds in $\gamma_0$, then we have $(pref_v=null)\wedge (\exists u\in N_v,married_u=false)$ by definition. Let us study the following cases.

\noindent\textbf{Case 3.1:} $u\in V_2$.\\
By Lemma~\ref{lemma6}, we know that $u$ is never activated in $e$. Consequently, the fairness of the daemon allows us to conclude that $pref_u=r\neq v$. Indeed, in the contrary case, $v$ is continuously enabled by $(M)$ if $pref_u=v$ and $u$ and $v$ are continuously enabled by $(M)$ or by $(S)$ if $pref_u=null$.

If $u$ is $doomed$, we can refer to case 2 with $u$ playing the role of $v$ while if $u$ is $proposing$, we can refer to case 1 with $u$ playing the role of $v$, that ends this case.

\noindent\textbf{Case 3.2:} $u\in V_1\setminus V_2$\\
Observe that $u$ cannot be $dead$ since $v$ is $single$. If $u$ is $single$, then $u$ is continuously enabled by $(S)$ or by $(M)$. If $u$ is $doomed$, then $u$ is enabled by $(A)$. If $u$ is $proposing$ after a potential bride $r$ different than $v$, then $r$ is continuously enabled by $(S)$ or $(M)$. The fairness of the daemon implies that $r$ is activated in a finite time and hence that $u$ remains $proposing$ only a finite time.

Consequently, we know that, while $u$ is not activated, remains not $married$, and is not $proposing$ after $v$, $u$ is infinitely often enabled. The fairness of the daemon implies that, while $u$ is not $married$ nor $proposing$ after $v$, $u$ is activated in a finite time. The construction of the algorithm and the round robin policy used for the management of the pointer ensure us that $u$ is either $married$ or $proposing$ after $v$ in a finite time. 

If $u$ becomes $proposing$ after $v$, then $v$ becomes enabled by $(M)$ and $u$ is not enabled while $v$ is not activated. Hence, the fairness of the daemon leads to an activation of $v$ in a finite time, that is contradictory.

If $u$ becomes $married$, then $v$ can remain $single$ or become $dead$. The first case allows us to refer to case 3 again (but this case can arise only a finite number of times since the number of neighbors of $v$ is finite). In the second case, we obtain a contradiction since $P$ strictly decrease.

In any case, we obtain a contradiction in a finite time and we can deduce the lemma.
\end{proof}

This set of Lemmas allows us to conclude on the following results:
\begin{lemma}
\label{lemma9}
Any execution of $\mathcal{SSMM}$ reaches a configuration of $\mathcal{LC}_2$ in a finite time under the central fair daemon.
\end{lemma}

\begin{theorem}
\label{theorem1}
$\mathcal{SSMM}$ is a $(2,n)$-strictly stabilizing protocol for $spec$ under the central fair daemon.
\end{theorem}

\subsection{Optimality of Containment Radius}

This section is devoted to the impossibility result that proves the optimality of the containment radius performed by $\mathcal{SSMM}$.

\begin{theorem}
\label{theorem2}
There exists no $(1,1)$-strictly stabilizing protocol for $spec$ under any daemon.
\end{theorem}

\begin{proof}
Consider a Byzantine city reduced to a chain of 5 potential brides labelled from left to right by $v_0$, $v_1$, ..., $v_4$. Consider the configuration $\gamma$ in which $v_0$ (resp. $v_3$) is married with $v_1$ (resp. $v_4$). Hence, $v_2$ is dead. Observe that $\gamma$ belongs to $\mathcal{LC}_1$ if the potential bride $v_0$ is Byzantine-faulty (\emph{i.e} any potential bride of $V_1$ is either $married$ or $dead$). 

By definition, any $(1,1)$-strictly stabilizing protocol for $spec$ must ensure the closure of $\mathcal{LC}_1$ for any execution starting from $\gamma$. But we can observe that it is not the case. Indeed, it is sufficient that the Byzantine-faulty potential bride breaks its marriage with $v_1$ during the first step for violating the closure of $\mathcal{LC}_1$ (since $v_2\in V_1$ becomes $single$). As no protocol can prevent a Byzantine fault by definition, we have the result.
\end{proof}

\section{Related Works}

Self-stabilization~\cite{D74j,D00b,T09bc} is a versatile technique that permits forward recovery from any kind of transient faults, while Byzantine fault-tolerance~\cite{LSP82j} is traditionally used to mask the effect of a limited number of malicious faults. In the context of self-stabilization, the first algorithm for computing a maximal marriage was given by Hsu and Huang \cite{HH92j}. Goddard et al.~\cite{GHJS03c} later gave a synchronous self-stabilizing variant of Hsu and Huang's algorithm. Finally, Manne et al.~\cite{MMPT09j} gave an algorithm for computing a maximal marriage under the distributed daemon. When it comes to improving the $\frac{1}{2}$-approximation induced by the maximal mariage property, Ghosh et al.~\cite{GGHSP95c} and Blair and Manne \cite{BM03c} presented a framework that can be used for computing a maximum mariage in a tree, while Goddard et al.~\cite{GHS06c} gave a self-stabilizing algorithm for computing a $\frac{2}{3}$-approximation in anonymous rings of length not divisible by three. Manne et al. later generalized this result to any arbitrary topology \cite{MMPT11j}. Note that contrary to our proposal, none of the aforementioned marriage construction algorithms can tolerate Byzantine behaviour. 

Making distributed systems tolerant to both transient and malicious faults is appealing yet proved difficult~\cite{DW04j,DD05c} as impossibility results are expected in many cases (even with complete communication topology and in a synchronous setting). A promising path towards multi-tolerance to both transient and Byzantine faults is Byzantine containment. For local tasks (\emph{i.e.} tasks whose correctness can be checked locally, such as vertex coloring, link coloring, or dining philosophers), strict stabilization~\cite{NA02c,MT07j} permits to contain the influence of malicious behavior to a fixed radius. This notion was further generalized for global tasks (such as spanning tree construction) using the notion of topology-aware strict stabilization \cite{DMT10ca,DMT10cd}. Our proposal is a strictly stabilizing maximal marriage protocol that has optimal containement radius.

\section{Conclusion}

We investigated the problem of recovering a catastrophic war by establishing long standing marriages, despite starting from an arbitrarily devastated state and having traitors trying make the global process fail. We presented evidence that no protocol can be completely resilient to traitors (as far as their influence containment is concerned), and designed and formally proved a protocol to solve the problem that is optimal in that respect. Further work is still needed for determining the global possible efficiency of the marriage process. It is known that in a scenario without traitors, a given maximal marriage~\cite{HH92j,MMPT09j} is a factor $2$ from the optimal (over all possible maximal marriages), yet more efficient solutions (with respect to the approximation ration) are possible~\cite{MMPT11j}. Extending those works to Byzantine-faulty setting is a challenging further work.

\bibliographystyle{plain}
\bibliography{biblio}

\end{document}